\documentclass[runningheads]{llncs}

\usepackage{latexsym}
\usepackage{amsfonts,amssymb,amsmath}
\usepackage{graphicx, gastex}
\usepackage{cite}
\usepackage{mathrsfs}

\DeclareSymbolFont{rsfscript}{OMS}{rsfs}{m}{n}
\DeclareSymbolFontAlphabet{\mathrsfs}{rsfscript}

\DeclareMathOperator{\Syn}{Syn}

\begin{document}
\title{Complexity of checking whether two automata are synchronized by the same language}
\author{Marina Maslennikova\thanks{The author acknowledges support from the Presidential Programm for young researchers, grant MK-3160.2014.1.}}
\titlerunning{Complexity of checking equality of languages of reset words}

\institute{Ural Federal University, Ekaterinburg, Russia\\
\email{maslennikova.marina@gmail.com}}

\maketitle

\begin{abstract}
A deterministic finite automaton $\mathscr{A}$ is said to be \emph{synchronizing} if it has a \emph{reset} word, i.e. a word that brings all states of the automaton $\mathscr{A}$ to a particular one. We prove that it is a \textbf{PSPACE}-complete problem to check whether the language of reset words for a given automaton coincides with the language of reset words for some particular automaton.
\end{abstract}

\section*{Introduction}
Let $\mathscr{A}=\langle Q,\Sigma,\delta\rangle$ be a \textit{deterministic finite automaton} (DFA),
where $Q$ is the \textit{state set}, $\Sigma$ stands for the \textit{input alphabet},
and $\delta: Q\times\Sigma\rightarrow Q$ is the totally defined \textit{transition function} defining
the action of the letters in $\Sigma$ on $Q$. The function $\delta$ is extended uniquely to a function $Q\times \Sigma^*\rightarrow Q$, where $\Sigma^*$ stands for the free monoid over $\Sigma.$ The latter function is still denoted by $\delta.$
In the theory of formal languages the definition of a DFA usually includes the \textit{initial state} $q_0\in Q$ and the set $F\subseteq Q$ of \textit{terminal states}. We will use this definition when dealing with automata as devices for recognizing languages. A language $L\subseteq \Sigma^*$ is \emph{recognized} (or \emph{accepted}) by an automaton $\mathscr{A}=\langle Q, \Sigma, \delta, q_0, F\rangle$ if $L=\{w\in \Sigma^*\mid \delta(q_0,w)\in F\}.$ We denote by $L[\mathscr{A}]$ the language accepted by the automaton $\mathscr{A}.$

A DFA $\mathscr{A}=\langle Q,\Sigma,\delta\rangle$ is called
\emph{synchronizing} if there exists a word $w  \in \Sigma^{*}$ whose action leaves
the automaton in one particular state no matter at which state in $Q$ it is applied: $\delta(q,w)=\delta(q',w)$
for all $q, q' \in Q$. Any word $w$ with this property is said to be
\emph{reset} for the DFA $\mathscr{A}$.
For the last 50 years synchronizing automata received a great deal of attention. In 1964
\v{C}ern\'{y}
conjectured that
every synchronizing automaton with $n$ states possesses a reset word of length at most $(n-1)^2$.
Despite intensive efforts of researchers this conjecture still remains open.
For a brief introduction to the theory of synchronizing automata we refer the reader
to the recent surveys~\cite{Vo_Survey, Sa05}.

In the present paper we focus on some complexity aspects of the theory of synchronizing automata.
We denote by $\Syn(\mathrsfs{A})$ the language of reset words for a given automaton $\mathscr{A}$.
It is well known that $\Syn(\mathrsfs{A})$ is regular~\cite{Vo_Survey}.
Furthermore, it is an \emph{ideal}
in $\Sigma^*$, i.e. $\Syn(\mathrsfs{A})=\Sigma^{*}\Syn(\mathrsfs{A})\Sigma^{*}$.
On the other hand, every regular ideal language $L$ serves as the language of reset words for some automaton.
For instance, the minimal automaton
recognizing $L$ is synchronized exactly by $L$~\cite{SOFSEM}. Thus synchronizing automata can be considered as a special representation of an ideal language.
Effectiveness of such a representation was addressed in~\cite{SOFSEM}.
The \textit{reset complexity} $rc(L)$ of an ideal language $L$ is the minimal possible number
of states in a synchronizing automaton $\mathscr{A}$ such that $\Syn(\mathscr{A})=L$.
Every such automaton $\mathscr{A}$ is called a \textit{minimal synchronizing automaton} (for brevity, MSA).
Let $sc(L)$ be the number of states in the minimal automaton recognizing $L$.
For every ideal language $L$ we have  $rc(L) \leq sc(L)$~\cite{SOFSEM}.
Moreover, there are languages $L_n$ for every $n \geq 3$ such that $rc(L_n) = n$ and $sc(L_n) = 2^n - n$~\cite{SOFSEM}.
Thus the representation of an ideal language by means of a synchronizing automaton can be exponentially more succinct than the
``traditional'' representation via the minimal automaton.
However, no reasonable algorithm is known for computing an MSA of a given language.
One of the obstacles is that an MSA is not uniquely defined.
For instance, there is a language with at least two different
MSAs \cite{SOFSEM}.

Let $L$ be an ideal regular language over $\Sigma$ with $rc(L)=n$. The latter equality means that there exists some $n$-state DFA $\mathscr{B}$ such that $\Syn(\mathscr{B})=L$, and $\mathscr{B}$ is an MSA for $L.$ Now it is quite natural to ask the following question: how hard is it to verify the condition $\Syn(\mathscr{B})=L$? It is well known that the equality of the languages accepted by two given DFAs can be checked in polynomial of the size of automata time. However, the problem of checking the equality of the languages of reset words of two synchronizing DFAs turns out to be hard. Moreover, it is hard to check whether one particular ideal language serves as the language of reset words for a given synchronizing automaton. We state formally the SYN-EQUALITY problem:

--\emph{Input:} synchronizing automata $\mathscr{A}$ and $\mathscr{B}$.

--\emph{Question:} is $\Syn(\mathscr{A})=\Syn(\mathscr{B})$?

We prove that SYN-EQUALITY is a \textbf{PSPACE}-complete problem. Actually, we prove a stronger result, that it is a \textbf{PSPACE}-complete problem to check whether the language $\Syn(\mathscr{A})$ for a given automaton $\mathscr{A}$ coincides with the language $\Syn(\mathscr{B})$ for some particular automaton $\mathscr{B}.$
Also it is interesting to understand how hard is it to verify a strict inclusion $\Syn(\mathscr{A})\subsetneq\Syn(\mathscr{B}).$ We prove that it is not easier than to check the precise equality of the languages $\Syn(\mathscr{A})$ and $\Syn(\mathscr{B}).$ So the problem of constructing an MSA for a given ideal language is unlikely to be an easy task. Also we obtain that the problem of checking the inequality $rc(L)\leq \ell$, for a given positive integer number $\ell$, is \textbf{PSPACE}-complete. Here an ideal language $L$ is presented by a DFA, for which $L$ serves as the language of reset words. Actually, we prove that the problem of checking the equalities $rc(L)=1$ or $rc(L)$ is trivial, however it is a \textbf{PSPACE}-complete problem to verify whether $rc(L)=3.$

The paper is organized as follows. In Section 1 we introduce some definitions and state formally the considered problems. In Section 2 we prove main results about \textbf{PSPACE}-completeness of the problem SYN-EQUALITY and \textbf{PSPACE}-completeness of the problem of checking whether the reset complexity of a given ideal language is not greater than $\ell$.

\section{Preliminaries}

A standard tool for finding the language of synchronizing words of a given DFA
$\mathrsfs{A}=\langle Q,\delta,\Sigma\rangle$ is the \emph{power automaton}
$\mathcal{P}(\mathrsfs{A})$. Its state set is the set $\mathcal{Q}$ of all nonempty subsets of $Q$, and the transition function is defined as a natural extension of $\delta$ on the set $\mathcal{Q}\times\Sigma$ (the resulting function is also denoted by $\delta$), namely, $\delta(S,a)=\{\delta(s,a)\mid s\in S\}$ for $S\subseteq Q$ and $a\in\Sigma$.
If we take  the set $Q$ as the initial state and singletons as final states in $\mathcal{P}(\mathscr{A})$, then we obtain an automaton recognizing $\Syn(\mathscr{A}).$
It is easy to see that if all the singletons identified to get a unique sink state $s$ (i.e. $s$ is fixed by all letters in $\Sigma$), the resulting automaton still recognizes $\Syn(\mathrsfs{A})$. Throughout the paper the term \emph{power automaton} and the notation $\mathcal{P}(\mathscr{A})$ will refer to this modified version.

One may notice now that the problem SYN-EQUALITY can be solved by the following naive algorithm. Indeed, we construct the power automata $\mathcal{P}(\mathscr{A})$ and $\mathcal{P}(\mathscr{B})$ for DFAs $\mathscr{A}$ and $\mathscr{B}$. Now it remains to verify that automata $\mathcal{P}(\mathscr{B})$ and $\mathcal{P}(\mathscr{A})$ accept the same language. However, the automaton $\mathcal{P}(\mathscr{A})$ has $2^{n}-n$ states, where $n$ is the number of states in the DFA $\mathscr{A}$. So we cannot afford to construct directly the corresponding power automata. Now we state formally the SYN-INCLUSION problem. It will be shown that SYN-INCLUSION is in \textbf{PSPACE}.

SYN-INCLUSION

--\emph{Input:} synchronizing automata $\mathscr{A}$ and $\mathscr{B}.$

--\emph{Question:} is $\Syn(\mathscr{A})\subseteq\Syn(\mathscr{B})$?

Since SYN-INCLUSION belongs to the class \textbf{PSPACE}, we obtain that SYN-EQUALITY is in \textbf{PSPACE} as well. Further we prove that the SYN-EQUALITY problem is complete for the class \textbf{PSPACE}. Now it is interesting to consider the SYN-STRICT-INCLUSION problem:

--\emph{Input:} synchronizing automata $\mathscr{A}$ and $\mathscr{B}.$

--\emph{Question:} is $\Syn(\mathscr{A})\subsetneq\Syn(\mathscr{B})$?

It will be shown that SYN-STRICT-INCLUSION is a \textbf{PSPACE}-complete problem.

Recall that the word $u\in \Sigma^*$ is a \emph{prefix} (\emph{suffix} or \emph{factor}, respectively) of the word $w$ if $w=us$ ($w=tu$ or $w=tus$, respectively) for some $t,s\in\Sigma^*$. A reset word $w$ for a DFA $\mathscr{A}$ is called \emph{minimal} if none of its proper prefixes nor suffixes is reset.
We will denote by $w[i]$ the $i^{th}$ letter of $w$ and by $|w|$ the length of the word $w$.
In what follows the word $w[i]w[i+1]...w[j]$, for $i<j$, will be denoted by $w[i..j]$.

\section{\textbf{PSPACE}-completeness}

\begin{theorem}\label{PSPACE}
SYN-INCLUSION is in \textbf{PSPACE}.
\end{theorem}
\begin{proof}
Savitch's theorem states that \textbf{PSPACE=NPSPACE} \cite{Savitch}. Therefore, it is enough to prove that SYN-INCLUSION belongs to \textbf{NPSPACE}, i.e. it suffices to solve the problem by a non-deterministic algorithm within polynomial space. Let $\mathscr{A}=\langle Q_1,\Sigma,\delta_1\rangle$ and $\mathscr{B}=\langle Q_2,\Sigma,\delta_2\rangle$ be synchronizing automata over $\Sigma$. We have to prove that the language $\Syn(\mathscr{B})$ contains the language $\Syn(\mathscr{A}),$ or equivalently the following equality takes place: $\Syn(\mathscr{A})\cap \Syn(\mathscr{B})^c=\emptyset,$ where $\Syn(\mathscr{B})^c$ is the complement language of $\Syn(\mathscr{B}).$ An obstacle is that we cannot afford to construct the automaton recognizing the language $\Syn(\mathscr{A})\cap \Syn(\mathscr{B})^c$ directly. Instead we provide an algorithm that guesses a word $w$ which is reset for $\mathscr{A}$ and is not reset for $\mathscr{B}.$ Let us notice that $w$ may turn out to be exponentially long in $|Q_1|+|Q_2|$. Hence even if our algorithm correctly guesses $w,$ it would not have enough space to store its guess. Thus the algorithm should guess $w$ letter by letter.

Let $Q_1=\{q_1,\ldots,q_n\}$ and $Q_2=\{p_1,\ldots,p_m\}.$ The algorithm guesses the first letter $w[1]$ of $w,$ applies $w[1]$ at every state in $Q_1$ and $Q_2$ and stores two lists of images, namely, $\{\delta_1(q_1,w[1]),$ $\ldots,$ $\delta_1(q_n,w[1])\}$ and $\{\delta_2(p_1,w[1]),$ $\ldots,$ $\delta_2(p_m,w[1])\}.$ These lists clearly require only $O(n+m)$ space. Further the algorithm guesses the second letter $w[2]$ and updates the lists of images re-using the space, and so on. Note that $\delta_1(q_i,w[1..k])=\delta_1(\delta_1(q_i,w[1..k-1]),w[k]),$ where $k\geq 2$ and $i\in\{1,\ldots,n\}.$ So we do not need to store the whole word $w$ in order to build the sets $\delta_1(Q_1,w)$ and $\delta_2(Q_2,w).$ At the end of the guessing steps the algorithm gets two lists $\{\delta_1(q_1,w),$ $\dots,$ $\delta_1(q_n,w)\}$ and $\{\delta_1(p_1,w),$ $\dots,$ $\delta_1(p_m,w)\}.$ It remains to check the following conditions:

\noindent-- all the states in the first list coincide with some particular state from $Q_1;$

\noindent-- there are at least two different states in the second list.

The latter checking does not require any additional space. Thus the problem SYN-INCLUSION is in \textbf{PSPACE.}
\qed
\end{proof}

Since $\Syn(\mathscr{A})=\Syn(\mathscr{B})$ if and only if $\Syn(\mathscr{A})\subseteq \Syn(\mathscr{B})$ and  $\Syn(\mathscr{B})\subseteq \Syn(\mathscr{A}),$ we obtain the following corollary.
\begin{corollary}\label{SYN-EQUALITY}
SYN-EQUALITY is in \textbf{PSPACE}.
\end{corollary}

To prove that SYN-EQUALITY is a \textbf{PSPACE}-complete problem we reduce the following well-known \textbf{PSPACE}-complete problem to the complement of SYN-EQUALITY. This problem deals with checking emptiness of the intersection of languages accepted by DFAs from a given collection \cite{Kozen}.

FINITE AUTOMATA INTERSECTION

--\emph{Input:} given $n$ DFAs $M_i=\langle Q_i,\Sigma,\delta_i,q_i,F_i\rangle,$ for $i=1,\ldots,n.$

--\emph{Question:} is $\bigcap_i L[M_i]\neq\emptyset?$

Given an instance of FINITE AUTOMATA INTERSECTION, we can suppose without loss of generality that each initial state $q_i$ has no incoming edges and $q_i\not\in F_i.$ Indeed, excluding the case for which the empty word $\varepsilon$ is in $L[M_i]$ we can always build a DFA $M'_i=\langle Q'_i,\Sigma,\delta'_i,q'_i,F_i\rangle$, which recognizes the same language of $M_i$, such that the initial state $q'_i$ has no incoming edges. This can be easily achieved by adding a new initial state $q'_i$ to the state set $Q_i$ and defining the transition function $\delta'_i$ by the rule: $\delta'_i(q'_i,a)=\delta_i(q_i,a)$ for all $a\in\Sigma$ and $\delta'_i(q,a)=\delta_i(q,a)$ for all $a\in\Sigma,$ $q\in Q_i.$ Furthermore, we may assume that the sets $Q_i$, for $i=1,\ldots,n,$ are pairwise disjoint.

To build an instance of SYN-EQUALITY from DFAs $M_i$, $i=1,\ldots,n,$ we construct a DFA $\mathscr{A}=\langle Q,\Delta,\varphi,\rangle$ with $Q=\bigcup^n_{i=1}Q_i \cup \{s,h\}$, where $s$ and $h$ are new states not belonging to any $Q_i.$
We add three new letters to the alphabet $\Sigma$ and get in this way $\Delta=\Sigma \cup \{x,y,z\}$.
The transition function $\varphi$ of the DFA $\mathscr{A}$ is defined by the following rules:\\
\begin{align*}
\varphi(q,a)&=\delta_i(q,a) &\text{for all }&i=1,\ldots,n,\text{for all }q\in Q_i \text{ and }a\in\Sigma;\\
\varphi(q,x)&=q_i &\text{for all }&i=1,\ldots,n,\text{for all }q\in Q_i;\\
\varphi(q,z)&=s &\text{for all }&i=1,\ldots,n,\text{for all }q\in F_i;\\
\varphi(q,z)&=h &\text{for all }&i=1,\ldots,n,\text{for all }q\in (Q_i\setminus F_i);\\
\varphi(q,y)&=s &\text{for all }&i=1,\ldots,n,\text{for all }q\in Q_i;\\
\varphi(h,a)&=s &\text{for all }&a\in\Delta;\\
\varphi(s,a)&=s &\text{for all }&a\in\Delta.
\end{align*}

The resulting automaton $\mathscr{A}$ is shown schematically in the Fig. 1. The action of letters from $\Sigma$ on the states $p\in Q_i$ is not shown. Denote by $G_i$ the set $Q_i\setminus(F_i\cup\{q_i\}).$ All the states from the set $G_i$ are shown as the node labeled by $G_i.$ All the states from the set $F_i$ are shown as the node labeled by $F_i.$

\begin{figure}[th]
\begin{center}
\unitlength=4pt
\begin{picture}(40,25)
    \gasset{Nw=4,Nh=4,Nmr=2}
    \thinlines
\node(h)(0,0){$h$}
\node(s)(40,0){$s$}
\node[Nframe=n](A3)(30,20){$\ldots$}
\node[Nframe=n](A2)(10,20){$\ldots$}
\node(q_1)(0,20){$q_i$}
\gasset{Nadjust=w,Nadjustdist=1.5,Nh=6,Nmr=2}
\node(p_1)(20,20){$G_i$}
\node(p_2)(40,20){$F_i$}
\drawloop[loopangle=0,loopdiam=4](s){$\Delta$}
\drawloop[loopangle=90,loopdiam=4](q_1){$x$}
\drawedge(q_1,h){$z$}
\drawedge(q_1,s){$y$}
\drawedge(h,s){$\Delta$}
\drawedge(p_1,h){$z$}
\drawedge(p_2,s){$y,z$}
\drawedge[ELside=r,curvedepth=-3](p_1,q_1){$x$}
\drawedge[ELside=r,curvedepth=-7](p_2,q_1){$x$}
\drawedge(p_1,s){$y$}
\drawedge(p_2,s){$y$}
\end{picture}
\caption{Automaton $\mathscr{A}$}
\end{center}
\label{FigA}
\end{figure}
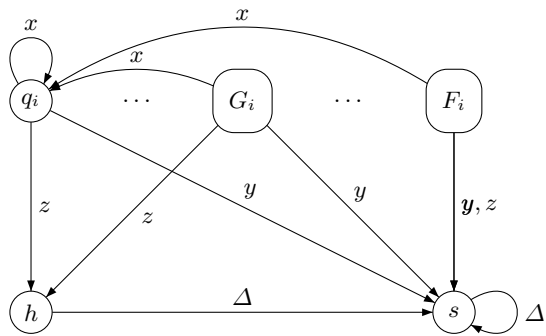

It can be easily seen that by the definition of the transition function $\varphi$ we get $\varphi(Q,w)\cap Q_i\neq\emptyset$ if and only if $w\in(\Sigma\cup\{x\})^*.$
From this observation and the definition of $\varphi$ we obtain the following lemma.

\begin{lemma}\label{PhiA}
For any $w\in\Delta^*$ we have $\varphi(Q,w)\cap Q_i\neq\emptyset$ for all $i=1,\ldots,n$ if and only if there is some $j\in\{1,\ldots,n\}$ such that $\varphi(Q,w)\cap Q_j\neq\emptyset.$
\end{lemma}

Consider the languages $L_1$ and $L_2$:\\
\begin{align*}
L_1&=(\Sigma\cup\{x\})^*y\Delta^*;\\
L_2&=(\Sigma\cup\{x\})^*z\Delta^+.
\end{align*}
Consider the language $I=L_1\cup L_2.$

\begin{lemma}\label{ResetA}
$\bigcap^n_{i=1}L[M_i]=\emptyset$ if and only if $\Syn(\mathscr{A})=I.$
\end{lemma}
\begin{proof}
Let $\bigcap^n_{i=1}L[M_i]=\emptyset$. We take a word $w\in\Syn(\mathscr{A}).$ Since $s$ is a sink state in $\mathscr{A},$ we have $\varphi(Q,w)=\{s\}.$ If $w\in(\Sigma\cup\{x\})^+,$ then $\varphi(Q,w)\cap Q_i\neq\emptyset,$ which is a contradiction. Therefore, $w$ contains some factor belonging to $\{y,z\}^+.$ Thus we can factorize $w$ as $w=uav$, where $u$ is a maximal prefix of $w$ belonging to $(\Sigma\cup\{x\})^*,$ $a\in\{y,z\}$ and $v\in\Delta^*.$

\textbf{Case 1}: $a=y.$

Since $y$ maps all the states of $\mathscr{A}$ to a sink state $s,$ we have that $w$ is a reset word and $w\in L_1.$

\textbf{Case 2}: $a=z.$

2.1. Let $u\in\Sigma^*,$ i.e. $u$ does not contain a factor belonging to $\{x\}^+.$ By lemma~\ref{PhiA}, $\varphi(Q,u)\cap Q_i\neq \emptyset$ for all $i=1,\ldots,n.$ Note that $u\not\in\bigcap_i L_i,$ that is $u\not\in L_j$ for some $j.$ It means that $\varphi(q_j,u)\in Q_j\setminus F_j.$ Hence $h\in \varphi(Q,uz).$ More precisely, $\varphi(Q,uz)=\{h,s\}.$ It remains to apply a letter from $\Delta$ in order to map $h$ to $s.$ So we have $w\in L_2.$

2.2. Let $u$ contain a factor belonging to $\{x\}^+.$ Obviously, we may factorize $u$ as $u=u'xt,$ where $t$ is the maximal suffix of $u$ belonging to $\Sigma^*.$ By lemma~\ref{PhiA}, $\varphi(Q,u')\cap Q_i\neq\emptyset$ for all $i=1,\ldots,n.$ Thus $q_i\in\varphi(Q,u'x)$ for all $i=1,\ldots,n.$ By the argument from the previous case, $\varphi(q_j,u'xt)\in Q_j\setminus F_j$ for some index $j.$ Thus $\varphi(Q,u'xtz)=\{h,s\}$ and $w\in L_2.$

So we obtain that if $\bigcap^n_{i=1}L[M_i]=\emptyset,$ then $\Syn(\mathscr{A})\subseteq I.$ The opposite inclusion $I\subseteq\Syn(\mathscr{A})$ follows easily from the arguments above. Assume now that the equality $\Syn(\mathscr{A})=I$ takes place. Arguing by contradiction, assume that $\bigcap^n_{i=1}L[M_i]\neq\emptyset,$ thus there is some $w'\in\bigcap^n_{i=1}L[M_i].$ By the definition of $\varphi$ we get that the word $w=xw'z$ is a reset word for $\mathscr{A}.$ However, $w\not\in I.$ So we came to a contradiction.
\qed
\end{proof}

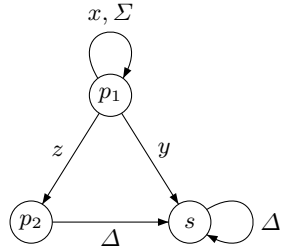
\begin{figure}[th]
\begin{center}
\unitlength=4pt
\begin{picture}(15,15)
    \gasset{Nw=4,Nh=4,Nmr=2,loopdiam=4}
    \thinlines
\node(s)(15,0){$s$}
\node(p_2)(0,0){$p_2$}
\node(p_1)(7.5,12){$p_1$}
\drawloop[loopangle=90](p_1){$x,\Sigma$}
\drawloop[loopangle=0](s){$\Delta$}
\drawedge(p_1,s){$y$}
\drawedge[ELside=r](p_2,s){$\Delta$}
\drawedge[ELside=r](p_1,p_2){$z$}
\end{picture}
\caption{Automaton $\mathscr{B}$}
\end{center}
\label{FigB}
\end{figure}

Now we build a $3$-state automaton $\mathscr{B}=\langle P,\Delta,\tau\rangle$ (see Fig.~\ref{FigB}). Its state set is $P=\{p_1,p_2,s\},$ where $s$ is the unique sink state. Further we verify that $I$ serves as the language of reset words for $\mathscr{B}.$ Moreover, $\mathscr{B}$ is an MSA for $I.$

\begin{lemma}\label{ResetB}
$\Syn(\mathscr{B})=I.$
\end{lemma}
\begin{proof}
It is clear that $I\subseteq \Syn(\mathscr{B}).$ Let $w\in\Syn(\mathscr{B}).$ Obviously, $w\not\in(\Sigma\cup\{x\})^*.$ Thus we may factorize $w$ as $w=uav,$ where $(\Sigma\cup\{x\})^*,$ $a\in\{y,z\}$ and $v\in\Delta^*.$ If $a=y$ then $w\in L_1.$ If $a=z$ then $\tau(Q,uz)=\{p_2,s\}.$ Since $w$ is a reset word for $\mathscr{B}$ we obtain that $w\in L_2.$ So we have the inclusion $\Syn(\mathscr{B})\subseteq I.$
\qed
\end{proof}

For each instance of FINITE AUTOMATA INTERSECTION one may construct the corresponding automaton $\mathscr{A}$ and the DFA $\mathscr{B}.$ It is easy to check that $I$ does not serve as the language of reset words for a synchronizing automaton of size at most two over the same alphabet $\Delta.$ So $\mathscr{B}$ is an MSA for $I$ and $rc(\Syn(\mathscr{B}))=3.$ Furthermore, $\mathscr{B}$ is a finitely generated synchronizing automaton, that is its language of reset words $\Syn(\mathscr{B})$ can be represented as $\Syn(\mathscr{B})=\Delta^*U\Delta^*$ for some finite set of words $U.$ Namely, $U=y\cup z\Delta.$ Finitely generated synchronizing automata and its languages of reset words were studied in \cite{FinGen,PribR1,PribRoPSPACE}. In particular, it was shown in \cite{PribRoPSPACE} that recognizing finitely generated synchronizing automata is a \textbf{PSPACE}-complete problem.
Finally, by lemmas~\ref{ResetA} and~\ref{ResetB}, we have the following claim.
\begin{lemma}\label{SeparateAB}
$\bigcap^n_{i=1}L[M_i]=\emptyset$ if and only if $\Syn(\mathscr{A})=\Syn(\mathscr{B}).$
\end{lemma}

Now we are in position to prove the main result.

\begin{theorem}\label{EqualityPsp}
SYN-EQUALITY is \textbf{PSPACE}-complete.
\end{theorem}
\begin{proof}
By Corollary~\ref{SYN-EQUALITY}, SYN-EQUALITY is in \textbf{PSPACE}. Since the construction of automata $\mathscr{A}$ and $\mathscr{B}$ can be performed in polynomial time from the automata $M_i$ ($i=1,\ldots,n$), by Lemmas~\ref{ResetA},~\ref{ResetB} and~\ref{SeparateAB}, we can reduce FINITE AUTOMATA INTERSECTION to co-SYN-EQUALITY.
\qed
\end{proof}

\begin{theorem}\label{InclusionPsp}
SYN-STRICT-INCLUSION is \textbf{PSPACE}-complete.
\end{theorem}
\begin{proof}
By theorem~\ref{PSPACE}, SYN-STRICT-INCLUSION is in \textbf{PSPACE}. We show that SYN-STRICT-INCLUSION is reduced to the SYN-EQUALITY problem and vice versa. Let $\mathscr{A}=\langle Q_1,\Sigma,\delta_1\rangle$ and $\mathscr{B}=\langle Q_2,\Sigma,\delta_2\rangle$ be synchronizing DFAs with $\Syn(\mathscr{A})=L$ and $\Syn(\mathscr{B})=L'.$ Define the automaton $\mathscr{A}\times\mathscr{B}=\langle Q,\Sigma,\delta\rangle$ in the following way:
\begin{align*}
Q&=Q_1\times Q_2;\\
\delta((q_1,q_2),a)&=(\delta_1(q_1,a),\delta_2(q_2,a))
\end{align*}
\noindent for all $(q_1,q_2)\in Q$ and $a\in\Sigma.$ Clearly, $\Syn(\mathscr{A}\times\mathscr{B})=L\cap L'.$ Thus $L\subsetneq L'$ if and only if $L=L\cap L'$ and $L\neq L'.$ On the other hand, $L\neq L'$ if and only if either $L\cap L'\subsetneq L$ or $L\cap L'\subsetneq L'.$
\qed
\end{proof}

Let us note that we build synchronizing automata $\mathscr{A}$ and $\mathscr{B}$ over at least $5$-letter alphabet to obtain an instance of SYN-EQUALITY. What about alphabets of size less than five? It can be easily seen that, for automata over a unary alphabet, SYN-EQUALITY can be solved in polynomial time. Indeed, if $\mathscr{A}=\langle Q_{1},\{a\},\delta_1\rangle$ is a synchronizing DFA, then $\Syn(\mathscr{A})=a^*a^k,$ where $k$ is the length of the shortest reset word for $\mathscr{A}.$ Furthermore, it is easy to check that $k<|Q_1|.$ Analogously, the language of reset words for a DFA $\mathscr{B}=\langle Q_{2},\{a\},\delta_2\rangle$ is $\Syn(\mathscr{B})=a^*a^m,$ where $m<|Q_2|.$ Finally, positive integer numbers $k$ and $m$ can be found in polynomial time. So it is interesting to consider automata over alphabets of size at least two.

We have reduced the problem FINITE AUTOMATA INTERSECTION to the problem SYN-EQUALITY.
By construction of DFAs $\mathscr{A}$ and $\mathscr{B}$, we have $\Delta=\{y,z,a,b,x\}.$ We build DFAs $\mathscr{C}=\langle C,\{\mu,\lambda\},\varphi_{2}\rangle$ and $\mathscr{D}=\langle D, \{\mu,\lambda\},\tau_2\rangle$ with unique sink states $\zeta_1$ and $\zeta_2$ respectively. It will be shown that $\Syn(\mathscr{A})=\Syn(\mathscr{B})$ if and only if $\Syn(\mathscr{C})=\Syn(\mathscr{D}).$ A standard technique is applied here and also was used in \cite{AnGuVol,Martygin,SOFSEM}. Namely, we define morphisms $h:\{\lambda,\mu\}^*\lambda\rightarrow\Delta^*$ and $\overline{h}:\Delta^*\rightarrow\{\lambda,\mu\}^*\lambda$ preserving the property of being a reset word for the corresponding automaton. Let $d_1=y$, $d_2=z,$ $d_3=a$, $d_4=b$ and $d_5=x$. We put $h(\mu^k\lambda)=d_{k+1}$ for $k=0,\ldots,4$ and $h(\mu^k\lambda)=d_5=x$ for $k\geq 5.$ Every word from the set $\{\lambda,\mu\}^*\lambda$ can be uniquely factorized by words $\mu^k\lambda,$ $k=0,1,2,\ldots,$ thus the mapping $h$ is totally defined. We also consider the morphism $\overline{h}:\Delta^*\rightarrow\{\lambda,\mu\}^*\lambda$ defined by the rule $\overline{h}(d_k)=\mu^{k-1}\lambda.$ Note that for every word $u\in\Delta^*$ we have $h(\overline{h}(u))=u.$

Now we take the constructed above DFA $\mathscr{B}=\langle P,\Delta,\tau\rangle$ with the state set $P=\{p_1,p_2,s\}.$ We build the DFA $\mathscr{D}=\langle D, \{\mu,\lambda\},\tau_2\rangle$ with a unique sink state $\zeta_2.$

We associate each state $p_i$ of the automaton $\mathscr{B}$ with the $5$-element set of states $P_i=\{p_{i,1},\ldots,p_{i,5}\}$ of the automaton $\mathscr{D}.$ Namely, the states $p_{i,2}, p_{i,3}, p_{i,4}, p_{i,5}$ are copies of the state $p_i$ associated with $p_{i,1}.$ The action of the letter $\mu$ is defined in the following way: $\tau_2(p_{i,k},\mu)=p_{i,k+1}$ for $k\leq4,$ and $\tau_2(p_{i,5},\mu)=p_{i,5}.$ We put $D=P_1\cup P_2\cup\{\zeta_2\},$ where $\zeta_2$ is a unique sink state. The action of the letter $\lambda$ is defined by the rules:

\noindent-- if $\tau(p_i,d_k)=s$, then $\tau_2(p_{i,k},\lambda)=\zeta_2;$

\noindent-- if $\tau(p_i,d_k)=p_j$, then $\tau_2(p_{i,k},\lambda)=p_{j,1}.$

The latter rule means that if there is the transition from $p_i$ to $p_j$ labeled by the letter $d_k,$ then there is the transition from $p_{i,1}$ to $p_{j,1}$ labeled by the word $\mu^{k-1}\lambda.$

$P_i$ is called the $i$\emph{-th colomn} of the set $D.$
For each $k=1,\ldots,5,$ one may take the set $R_k=\{p_{1,k},p_{2,k}\}.$ The set $R_k$ is called the $k$-\emph{th row} of the set $D.$

The DFA $\mathscr{C}$ is constructed in analogous way. Finally, note that the resulting automata $\mathscr{C}$ and $\mathscr{D}$ have $O(5|Q_1|)$ and $O(5|Q_2|)$ states respectively, where $|Q_1|$ and $|Q_2|$ are the cardinalities of the state sets of $\mathscr{A}$ and $\mathscr{B}$ respectively.
Figure~3 illustrates the automaton $\mathscr{D}.$ The action of the letter $\mu$ is shown in solid lines, the action of the letter $\lambda$ is shown in dotted lines.
\begin{figure}[th]
\begin{center}
\unitlength=4pt
\begin{picture}(20,55)
    \gasset{Nw=4,Nh=4,Nmr=2,loopdiam=4}
    \thinlines
\node(11)(5,50){$p_{1,1}$}
\node(12)(5,40){$p_{1,2}$}
\node(13)(5,30){$p_{1,3}$}
\node(14)(5,20){$p_{1,4}$}
\node(15)(5,10){$p_{1,5}$}
\node(21)(15,50){$p_{2,1}$}
\node(22)(15,40){$p_{2,2}$}
\node(23)(15,30){$p_{2,3}$}
\node(24)(15,20){$p_{2,4}$}
\node(25)(15,10){$p_{2,5}$}
\node(s)(25,50){$\zeta_2$}
\drawloop[loopangle=-45](s){}
\drawloop[loopangle=-90](25){}
\drawloop[loopangle=-90](15){}
\drawedge(11,12){}
\drawedge(12,13){}
\drawedge(13,14){}
\drawedge(14,15){}
\drawedge(21,22){}
\drawedge(22,23){}
\drawedge(23,24){}
\drawedge(24,25){}
\drawedge[dash={0.5}0](21,s){}
\drawedge[dash={0.5}0](22,s){}
\drawedge[dash={0.5}0](23,s){}
\drawedge[dash={0.5}0](24,s){}
\drawedge[dash={0.5}0](25,s){}
\drawedge[dash={0.5}0](12,21){}
\drawloop[loopangle=45,dash={0.5}0](s){}
\drawedge[curvedepth=3,dash={0.5}0](13,11){}
\drawedge[curvedepth=5,dash={0.5}0](14,11){}
\drawedge[curvedepth=7,dash={0.5}0](15,11){}
\drawedge[curvedepth=3,dash={0.5}0](11,s){}
\put(-10,50){$y$}
\put(-10,40){$z$}
\put(-10,30){$a$}
\put(-10,20){$b$}
\put(-10,10){$x$}
\end{picture}
\caption{Automaton $\mathscr{D}$}
\end{center}
\label{FigD}
\end{figure}
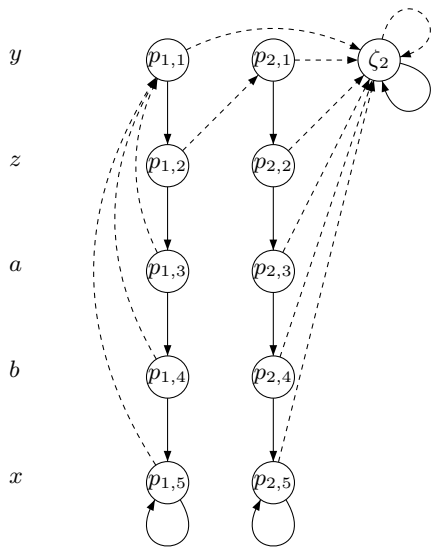

\begin{lemma}\label{ResetCD}
$\Syn(\mathscr{A})=\Syn(\mathscr{B})$ if and only if $\Syn(\mathscr{C})=\Syn(\mathscr{D}).$
\end{lemma}

\begin{proof}
It is convenient to organize the constructed DFA $\mathscr{D}$ as a table. The $k$-th row contains copies of all states corresponding to the $k$-th letter from $\Delta.$ The $i$-th column contains the state $p_{i,1}$ corresponding to the state $p_i$  and its copies $p_{i,2},\ldots,p_{i,5}.$ Each state from the $i$-th column maps under the action of $\mu$ to a state from the same column. The $k$-th row $R_k$ maps under the action of $\mu$ to the $k+1$-st row $R_{k+1}$ ($k\leq 4$). The 5-th row is fixed by $\mu$, that is $\tau_2(R_5,\mu)=R_5.$ The state set $D$ maps under the action of $\lambda$ to a subset of $R_1.$
The DFA $\mathscr{C}$ possesses such properties as well.

It can be easily checked that the word $u\in\Delta^*$ is reset for the DFA $\mathscr{B}$ if and only if $\overline{h}(u)$ is reset for $\mathscr{D}.$ An analogous property takes place for DFAs $\mathscr{A}$ and $\mathscr{C}.$

Assume that $\Syn(\mathscr{A})\neq\Syn(\mathscr{B}).$ From the proof of lemma~\ref{ResetA} it follows that the word $w=xw'z$ with $w'\in\bigcap_i L[M_i]$ is reset for $\mathscr{A}$ and it is not reset for $\mathscr{B}.$ Thus $\overline{h}(w)\in\Syn(\mathscr{C})$ and $\overline{h}(w)\not\in\Syn(\mathscr{D}).$ So $\Syn(\mathscr{C})\neq\Syn(\mathscr{D}).$

Assume now that $\Syn(\mathscr{A})=\Syn(\mathscr{B}).$ We show that every minimal reset word of $\mathscr{C}$ is reset for $\mathscr{D}$ and every minimal reset word of $\mathscr{D}$ is reset for $\mathscr{A}.$ Let $u$ be a minimal reset word of $\mathscr{C}.$ Any word $u\in\{\mu\}^*$ is not in $\Syn(\mathscr{C}),$ since $\mu$ brings each column to its subset. Thus we have $u\in\{\lambda,\mu\}^*\setminus\{\mu\}^*.$ The automaton $\mathscr{C}$ possesses a unique sink state $\zeta_1.$ Hence $\mathscr{C}$ is synchronized to $\zeta_1.$ Furthermore, all the transitions leading to $\zeta_1$ are labeled by $\lambda,$ and $\zeta_1$ is fixed by $\mu$ and $\lambda.$ Thus if $u$ does not end with $\lambda$ then it is not a minimal reset word. We have $u\in\{\lambda,\mu\}^*\lambda.$ Consider the word $w=h(u).$ Since $\overline{h}(w)=u,$ we have that $w$ is a reset word for $\mathscr{A}$ and $\mathscr{B}.$ Hence $u\in\Syn(\mathscr{D}).$ So we obtain that $\Syn(\mathscr{C})\subseteq \Syn(\mathscr{D}).$ The opposite inclusion is verified analogously.
\qed
\end{proof}

Lemma~\ref{ResetCD} gives the desired result on \textbf{PSPACE}-completeness of the problem SYN-EQUALITY restricted to a binary alphabet case.

\begin{proposition}\label{rcPSPACE}
Let $\ell$ be a positive integer number, $L$ an ideal language, and $\mathscr{A}$ a synchronizing DFA for which $L$ serves as the language of reset words. The problem of checking the inequality $rc(L)\leq\ell$ is in \textbf{PSPACE}.
\end{proposition}
\begin{proof}
If $\ell$ is greater or equal to the size of $\mathscr{A},$ then the answer is ``yes'' and there is nothing to prove. Let $\ell$ be less than the size of $\mathscr{A}.$ One may non-deterministically guess a DFA $\mathscr{B}$ with at most $\ell$ states and check the equality $\Syn(\mathscr{B})=\Syn(\mathscr{A})$ within polynomial space.
\qed
\end{proof}

\begin{lemma}
Let $L$ be an ideal language and $\mathscr{A}$ some automaton with $\Syn(\mathscr{A})=L.$ The equalities $rc(L)=1$ and $rc(L)=2$ can be checked in polynomial of the size of $\mathscr{A}$ time.
\end{lemma}
\begin{proof}
Let $\mathscr{A}=\langle Q,\Sigma,\delta\rangle.$ Denote by $k$ the size of the alphabet $\Sigma.$ It is easy to see that $rc(L)=1$ if and only if $L=\Sigma^*,$ so it is required that $n=1.$

Let us notice that some $2$-state automaton $\mathscr{B}=\langle P,\Sigma,\delta\rangle$ is synchronizing if and only if some letter brings the automaton to a singleton and each letter $a\in\Sigma$ either maps the state set $P$ to a singleton or acts as a permutation on $P.$ So we find the set $\Gamma=\{a\mid a\in \Syn(\mathscr{A})\}$ in time $O(kn)$ and obtain the DFA $\mathscr{A}'=\langle Q,\Sigma\setminus\Gamma,\delta\rangle$ from $\mathscr{A}$ removing the transitions labeled by letters from $\Gamma.$ It remains to check that $\mathscr{A}'$ is not synchronizing. The latter checking can be done in time $O(kn^2)$ \cite{Epp}. We have that $rc(L)=2$ if and only if $rc(L)\neq 1$ and $\mathscr{A}'$ is not synchronizing.
\qed
\end{proof}

We have constructed for each instance of FINITE AUTOMATA INTERSECTION the corresponding automaton $\mathscr{A}$  over the alphabet $\Delta=\Sigma\cup \{x,y,z\}$ in order to prove that SYN-EQUALITY is a \textbf{PSPACE}-complete problem.
We will use that automaton again to prove the following theorem.

\begin{theorem}
Let $J=\Syn(\mathscr{A}).$ $\bigcap^n_{i=1}L[M_i]\neq\emptyset$ if and only if $rc(J)>3.$ 
\end{theorem}
\begin{proof}

\noindent Let $\bigcap^n_{i=1}L[M_i]=\emptyset.$ In this case we have $\Syn(\mathscr{A})=I.$ As it was mentioned above $\mathscr{B}$ is an MSA for $I,$ so we have the equality $rc(\Syn(\mathscr{A}))=3$. Let us assume now that $\bigcap^n_{i=1}L[M_i]\neq\emptyset.$ Since $y$ is a unique reset letter for $\mathscr{A}$ and the automaton $\mathscr{A}'$ (which is obtained from $\mathscr{A}$ by removing of all transitions labeled by $y$) is still synchronizing, we get that $rc(J)\geq 3.$ Now we verify that $rc(J)\neq 3,$ that is no $3$-state synchronizing automaton is synchronized exactly by $J.$ It is easy to see that if $\bigcap^n_{i=1}L[M_i]\neq\emptyset,$ then $J=I\cup I_3\cup I_4$ where\\
\begin{align*}
I_3&=\{wz\mid w\in\bigcap^n_{i=1}L[M_i] \text{ and }\delta_i(Q_i,w)\subseteq F_i\};\\
I_4&=\{uxwz\mid u\in(\Sigma\cup\{x\})^*,w\in\bigcap^n_{i=1}L[M_i]\}.
\end{align*}

Arguing by contradiction, assume that $\Syn(\mathscr{B})=J$ for some $3$-state automaton $\mathscr{B}$ over $\Delta.$ Denote the state set of $\mathscr{B}$ by $P=\{0,1,2\}$ and the transition function by $\tau.$ Since $y$ is a reset letter for $\mathscr{B},$ we have that $y$ brings $P$ to a singleton, say $\{2\}.$ Letter $z$ is not reset for $\mathscr{A}$ and $z^2\in J,$ hence $z$ maps $P$ to $2$-element subset. It is easy to check that there are just three possible different ways of defining the action of $z$ on the state set $P$ (see Fig.4).

\begin{figure}[th]
\begin{center}
\unitlength=4pt
\begin{picture}(75,20)
    \gasset{Nw=4,Nh=4,Nmr=2,loopdiam=4}
    \thinlines
\node(s)(10,5){$2$}
\node(p_2)(-5,5){$0$}
\node(p_1)(2.5,17){$1$}
\drawloop[loopangle=-45](s){$y,z$}
\drawedge(p_1,s){$y$}
\drawedge[ELside=r](p_2,s){$y,z$}
\drawedge[ELside=r](p_1,p_2){$z$}
\node(2s)(40,5){$2$}
\node(2p_2)(25,5){$0$}
\node(2p_1)(32.5,17){$1$}
\drawloop[loopangle=0](2s){$y$}
\drawloop[loopangle=135](2p_2){$z$}
\drawedge(2p_1,2s){$y$}
\drawedge[ELside=r](2p_2,2s){$y$}
\drawedge[ELside=r](2p_1,2p_2){$z$}
\drawedge[curvedepth=3](2s,2p_1){$z$}
\node(3s)(75,5){$2$}
\node(3p_2)(60,5){$0$}
\node(3p_1)(67.5,17){$1$}
\drawloop[loopangle=0](3s){$y$}
\drawloop[loopangle=180](3p_2){$z$}
\drawedge[ELside=r,curvedepth=-3](3s,3p_2){$z$}
\drawedge(3p_1,3s){$y,z$}
\drawedge[ELside=r](3p_2,3s){$y$}
\end{picture}
\caption{Possible ways of defining the action of $y$ and $z$ in $\mathscr{B}$.}
\end{center}
\label{FigBrc}
\end{figure}
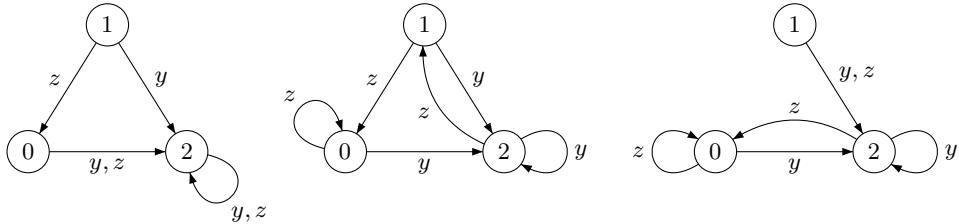

Let us assume as above that $\Sigma=\{a,b\}.$ It remains to define the action of $x,a$ and $b$ on the state set $P.$ One may see that the words $za,zb$ and $zx$ are reset for $\mathscr{A}.$ So letters $x,a$ and $b$ should map the set $\tau(P,z)$ to singletones. However, any word from $(\Sigma\cup\{x\})^*$ is not reset for $\mathscr{A}.$ In particular, $xx,aa,bb\not\in \Syn(\mathscr{A}).$ Consider, for instance, the first automaton in the Fig.4. We have that $\tau(P,z)=\{0,2\},$ thus the action of $x,$ $a$ and $b$ is defined in such way that $|\tau(\{0,2\},x)|=|\tau(\{0,2\},a)|=|\tau(\{0,2\},b)|=1.$ So there are six possible ways of defining the action of $x$ on the states of $\mathscr{B}.$ Since $xx\not\in\Syn(\mathscr{A}),$ we have that the following two ways of defining the transitions under the action of $x$ are impossible:
\begin{align*}
&012\\
x\text{ }&020\\
x\text{ }&202.
\end{align*}
Indeed, in both cases the word $xx$ brings the set $\{0,1,2\}$ to a singletone. So, actually, there are just four possible ways of defining the action of $x$ on the states of $\mathscr{B}.$ The same arguments can be provided for letters $a$ and $b.$ 
Thus the definition of the action of $x,$ $a,$ and $b$ is chosen in one of the following ways:
\begin{align*}
&012\\
x_1\text{ }&010\\
x_2\text{ }&212\\
x_3\text{ }&101\\
x_4\text{ }&121.
\end{align*}
For instance, one may say that $x$ acts on $P$ as $x_1,$ $a$ acts as $x_2$ and $b$ acts as $x_3.$ It is sufficient to consider only those cases where all the letters $x,$ $a$ and $b$ act on $P$ differently. There remains four ways of choosing a triple $\{x_i,x_j,x_k\}$ defining the action of letters $x,$ $a$ and $b.$ It can be easily checked that $J\neq \Syn(\mathscr{B})$ in each case. Analogous arguments are provided for the rest two automata in the Fig.4.
It means that the reset complexity of the language of reset words of the DFA $\mathscr{A}$ is at least $4$, that is $rc(J)\geq 4.$
\qed
\end{proof}

\begin{corollary}
Let $L$ be and ideal language and $\mathscr{A}$ a synchronizing DFA over at least $5$-letter alphabet with $\Syn(\mathscr{A})=L$. The problem of checking the inequality $rc(L)\leq 3$ is \textbf{PSPACE}-complete.
\end{corollary}

\textbf{Acknowledgment}. The author is grateful to participants of the seminar ``Theoretical Computer Science''  for valuable comments.

\end{document}